\documentclass[12pt]{iopart}

\usepackage{iopams}
\usepackage{harvard}
\usepackage{bibtexlogo}
\usepackage{amsthm}

\newtheorem{theorem}{Theorem}[section]
\newtheorem{lemma}[theorem]{Lemma}
\newtheorem{define}[theorem]{Definition}

\begin{document}
\bibliographystyle{dcu}

\title[]{SIC-POVMs and the Knaster's Conjecture}

\author{S. B. Samuel and Z. Gedik}

\address{Faculty of Engineering and Natural Sciences, Sabanci University, Tuzla, Istanbul 34956, Turkey}
\ead{solomon@sabanciuniv.edu , gedik@sabanciuniv.edu}
\vspace{10pt}
\begin{indented}
\item[]April 2024
\end{indented}

\begin{abstract}
Symmetric Informationally Complete Positive Operator-Valued Measures (SIC-POVMs) have been constructed in many dimensions using the Weyl-Heisenberg group. In the quantum information community, it is commonly believed that SCI-POVMs exist in all dimensions; however, the general proof of their existence is still an open problem. The Bloch sphere representation of SIC-POVMs allows for a general geometric description of the set of operators, where they form the vertices of a regular simplex oriented based on a continuous function. We use this perspective of the SIC-POVMs to prove the Knaster's conjecture for the geometry of SIC-POVMs and prove the existence of a continuous family of generalized SIC-POVMs where $(n^2-1)$ of the matrices have the same value of $Tr(\rho^k)$. Furthermore, by using numerical methods, we show that a regular simplex can be constructed such that all its vertices map to the same value of $Tr(\rho^3)$ on the Bloch sphere of $3$ and $4$ dimensional Hilbert spaces. In the $3$-dimensional Hilbert space, we generate $10^4$ generalized SIC-POVMs for randomly chosen $Tr(\rho^3)$ values such that all the elements are equivalent up to unitary transformations. 
\end{abstract}

%
\vspace{2pc}

\section{Introduction}

The Symmetric Informationally Complete Positive Operator-Valued Measures, in short, the SIC-POVMs, are a measurement formalism in quantum mechanics with applications in state tomography \cite{scott2006tight}, quantum cryptography \cite{durt2008wigner} and the foundational study of QBism \cite{fuchs2010qbism}. In an $n$-dimensional Hilbert space, the SIC-POVMs are defined as a set of rank-$1$ projection operators $\{\Pi_k\}$ where $Tr(\Pi_j\Pi_k) = \frac{1}{n+1}$ and the corresponding probability of the $k^{th}$ measurement is $P_k=\frac{1}{n}Tr(\rho \Pi_k)$. In 1999, Gerhard Zauner conjectured that a special vector, known as a fiducial vector, exists such that its orbit on the Weyl-Heisenberg group forms a SIC-POVM \cite{zauner}. In the same work, Zauner conjectured a symmetry of fiducial vectors and constructed SIC-POVMs in dimensions 2-7. Renes \etal independently constructed Weyl-Heisenberg SIC-POVMs and established the connection of SIC-POVMs and spherical 2-design \cite{renes2004symmetric}. Since then, SIC-POVMs have been constructed in every dimension 2-151 and many higher dimensions as high as 39604, \cite{scott2010symmetric,scott2017sics,appleby2022sic,exactfiducialsolns,bengtsson2024sic}. The only other type of SIC-POVM known is the Hoggar SIC in dimension 8, which is constructed using the Handmard matrices \cite{hoggar1981two}. Despite the large number of SIC-POVMs constructed to date, the proof of Zauner's conjecture remains open. In this article, we approached the SIC-POVMs from a geometric point of view by mapping the SIC-POVMs onto the Bloch sphere, which for our purpose, is generated by using the generalized Gell-Mann matrices.

On the Bloch sphere, pure states correspond to points on the surface of the sphere that are the global maxima of a homogeneous polynomial equation derived from $Tr(\rho^3)$. Because SIC-POVMs can be represented by rank-$1$ Hermitian matrices, we can map each of the operators onto the Bloch sphere, where they form a regular simplex such that each vertex corresponds to a pure state. This allows us to define SIC-POVMs equivalently as a set of equiangular vectors on the Bloch sphere which map to the maximum value of the homogeneous polynomial. This way of defining the SIC-POVMs allows us to connect the SIC-POVMs and a geometric problem known as Knaster's conjecture.

In 1947, Bronislaw Knaster posed the following problem in the Colloquium Mathematicae \cite{CollMath1947}, which was originally in French and presented in English in the "Life and Work of Bronislaw Knaster" in \cite{RomanDuda1987} as follows. \textit{"Given three points $p_1,p_2,p_3$ on a $2$-dimensional sphere and a continuous mapping  of $S^2$ into the $\mathbb R$. Do there exist points $q_1,q_2,q_3$ on $S^2$ which are equivalent to rotation to $p_1,p_2,p_3$ that have the common image on $\mathbb R$, i.e., }

\begin{equation}
f(q_1)=f(q_2)=f(q_3)\mbox{ \textit{?}}
\end{equation}

\noindent
\textit{More generally, For any given $k$ points $(p_1,\cdots,p_k)$ on $S^n$ and any continuous mapping of $S^n$ into $\mathbb R^{n-k+2}$, where $k\in \{2,\cdots,n+1\}$, is there a set of points $(q_1,\cdots,q_k)$ which has the same image in $\mathbb R^{n-k+2}$?"} The problem has been answered for various special sets of vectors. One of the most notable examples is Kakutani's theorem. Kakutani proved the conjecture for orthogonal vectors on a $3$-dimensional sphere in \cite{kakutani1942proof}. Later, Yamebe and Yujobo (1950) generalized the theorem to orthogonal vectors in arbitrary dimensions. Other examples are the Borsuk-Ulam theorem, which corresponds to antipodal points on an $n$-sphere \cite{borsuk1933drei} and Dyson's theorem, which corresponds to $2$ orthogonal lines crossing through the origin of an $n$-sphere \cite{dyson1951continuous}. 

In section \ref{proof of conjecture}, by using the symmetry of the SIC-POVMs, we prove the Knaster's conjecture for a set of $N$ unit vectors $\{\vec p_k\in \mathbb R^N | \vec p_j \cdot \vec p_k = -\frac{1}{N}\}$. Following that, we prove that if a SIC-POVM exists in $n$-dimensional Hilbert space, then continuous generalized SIC-POVMs, which are defined in \cite{gour2014construction}, exist such that $n^2-1$ of the elements have a common trace value $Tr(\Pi_k^m) = f_0$ for all $m \geq 3$. In dimension 3, this indicates the existence of tight frames where all but one of the elements are equivalent up to unitary transformations. We further confirmed this numerically and discovered that all the $9$ vectors forming the tight frame can be constructed, such that all the elements have the same set of eigenvalues. These tight frames form a continuous set of generalized SIC-POVMs, which in general are not group covariant. In dimension 4, we found that 16 Hermitian matrices can be constructed such that all the matrices have the same $Tr(\rho^3)$ values, for all possible values of the function.  Similarly, for $Tr(\rho^4)$, 16 equiangular vectors are constructed such that all vectors map to arbitrary values of the function.

\section{The Bloch Sphere}

The Bloch sphere is the representation of $2$-dimensional density matrices in a $3$-dimensional real space \cite{Bloch1946}. All points on the surface of the Bloch sphere correspond to pure state density matrices and the interior of the Bloch sphere corresponds to mixed state density matrices. The map from the $2$-dimensional Hilbert space to the Bloch sphere is generated by the Pauli matrices, $(\sigma_x,\sigma_y,\sigma_z)$, which form the basis for the space of trace-less Hermitian matrices. A density matrix $\rho$ is mapped to a vector $\vec r$ in the Bloch sphere, such that $\rho = \frac{1}{2} \mathbb I + \vec \sigma \cdot \vec r$ or in its matrix form as

\begin{equation}\label{2d density matrix}
    \rho = 
    \left(\begin{array}{cc}
\frac{1}{2}+r_z & r_x-ir_y \\
r_x+ir_y & \frac{1}{2}-r_z
\end{array}\right) \mbox{ .}
\end{equation}

\noindent
Conversely, the vector $\vec r$ is computed as $\vec r(\rho) = \frac{1}{2} Tr(\rho \vec \sigma)$.

The Bloch sphere can be generalized to higher dimensions by using different basis matrices. For example in dimension 3, the Gell-Mann matrices, polarization operators, and the Weyl-Heisenberg group were explored as choices for the generalization of the Bloch sphere in \cite{Bertlmann2008}. For our purposes, we need the Bloch sphere to be in the real vector space for geometric simplicity. Thus, we will use the generalized Gell-Mann matrices to define the Bloch sphere in arbitrary dimensions. Additionally, the construction of the Gell-Mann matrices is simpler compared to other basis matrices which makes them preferable for computations. The Bloch sphere defined by using the generalized Gell-Mann matrices has been explored in \cite{kimura2003bloch}, where the equations defining the pure states and mixed states are presented. For completeness purposes, we will briefly describe the properties and characterization of density matrices.

In $n$-dimensional Hilbert space, the Gell-Mann matrices are a set of $n^2-1$ Hermitian matrices that span the space of trace-less Hermitian matrices. Following the notation in \cite{Bertlmann2008}, we write the generalized Gell-Mann matrices as .

\begin{equation}\label{generalised gellmann matrices}
    \begin{array}{lll}
    \Lambda_{sym}^{jk} &= |j\rangle\langle k|+|j\rangle\langle k| &, 1\leq j<k \leq n \mbox{ ,}\\
    \Lambda_{asym}^{jk} &= -i(|j\rangle\langle k|-|j\rangle\langle k|) &, 1\leq j<k \leq n \mbox{ ,} \\
    \Lambda_{diag}^l &= \frac{2}{\sqrt{l(l+1)}} \sum_{k=1}^{l} {|k\rangle\langle k|-l|l+1\rangle\langle l+1|} &, 0\leq l \leq n-1 \mbox{ .}
    \end{array}
\end{equation}

\noindent
Like the Pauli matrices, generalized Gell-Mann matrices are mutually orthogonal, i.e., $Tr(\Lambda_j \Lambda_k) =2 \delta_{jk}$. To simplify the notations, we will represent the generalized Gell-Mann matrices as a vector $\vec \Lambda$ by placing the matrices in the same order as \eref{generalised gellmann matrices}. Then, a density matrix $\rho$ in $n$-dimensional Hilbert space is mapped onto the Bloch sphere by the equality $\vec r(\rho) = \frac{1}{2} Tr(\rho \vec \Lambda)$, such that $\rho = \frac{1}{n}\mathbb I + \vec r \cdot \vec \Lambda$. 

To explore the geometrical properties of quantum states on the Bloch sphere, in addition to the vector representation of the quantum states, we need to map the unitary operators onto the Bloch sphere. Let $u$ be a unitary matrix, $\rho$ an arbitrary density matrix, and let  $\rho'$ be  $u\rho u^\dagger$. The corresponding vector $\vec r'$ can be computed by taking the trace of $\rho' \vec \Lambda$, 

\begin{equation}
\begin{array}{ll}
     \frac{1}{2}Tr(\rho' \Lambda_j) &= \frac{1}{2}Tr(u\rho u^\dagger)\\
          &= \frac{1}{2n}Tr(u\mathbb I u^\dagger\Lambda_j) + \frac{1}{2} Tr(u\Lambda_k u^\dagger\Lambda_j) r_k \mbox{ .}
\end{array}
\end{equation}

\noindent
As a result, $\vec r' = \frac{1}{2}Tr(\Lambda_ju\Lambda_k u^\dagger) r_k$. Thus, a unitary transformation is mapped to the Bloch sphere by the equation $M_{jk} = \frac{1}{2}Tr(\Lambda_ju\Lambda_k u^\dagger)$. The matrix $M$ is an $(n^2-1)\times (n^2-1)$  orthogonal matrix.

\subsection{Density Matrices on the Bloch Sphere}\label{density matrix on bloch}

Compared to higher dimensions, the Bloch sphere of a $2$-dimensional Hilbert space has the simplest geometry, since all the points on the surface of the sphere correspond to pure states and the interior of the sphere corresponds to mixed states. This is evident from the two-to-one homeomorphism between $SU(2)$ and $SO(3)$. In higher dimensions, the geometry of the quantum states is more complicated. To illustrate this point, consider the density matrix $\rho$ of an $n$-dimensional system, whose eigenvalues are $\{\lambda_1,\cdots,\lambda_n\}$. The eigenvalues are solutions to the characteristic polynomial of $\rho$ given by 

\begin{equation}\label{characteristic polynomial}
Det|\lambda \mathbb I - \rho| = \prod_{k=1}^n (\lambda-\lambda_k) \mbox{ .}
\end{equation}

\noindent
The coefficients of the polynomial are functions of $Tr(\rho^k)$ where $k\leq n$. This is clear once we expand the right side of  \eref{characteristic polynomial}, 

\begin{equation}\label{characteristic in trace}
\begin{array}{ll}
\prod_{k=1}^n (\lambda-\lambda_k) &= \lambda^n - \lambda^{n-1} \bigg( \sum_j \lambda_j\bigg) + \lambda^{n-2} \bigg( \sum_{j\neq k} \lambda_j\lambda_k \bigg)-\cdots\\& +(-1)^n \bigg(\prod_k \lambda_k \bigg)\\
&= \lambda^n - \lambda^{n-1} Tr(\rho)+\frac{1}{2}\lambda^{n-2}\bigg( Tr(\rho)^2-Tr(\rho^2)\bigg) - \cdots \\&+ (-1)^n Det|\rho| \mbox{ .}
\end{array} 
\end{equation}

\Eref{characteristic in trace} implies that a density matrix can be characterized up to a unitary equivalence by the values of $Tr(\rho^k)$ where $k\leq n$. The trace of powers of $\rho$ are the functions that will determine the geometry of quantum states on the Bloch sphere. In dimension $2$, we only have two relevant values, $Tr(\rho)$ and $Tr(\rho^2)$. By definition $Tr(\rho) =1$, and $Tr(\rho^2)$ determines the eigenvalues of the density matrix. In the Bloch sphere, $Tr(\rho^2)$ translates to $|\vec r|^2$, where $\vec r$ is the image of the density matrix on the Bloch sphere. In general, for $n$-dimensional Hilbert space, we have to consider $n$ trace values which generate $n$ polynomial functions on the Bloch sphere. These polynomials can be generated by replacing the density matrix $\rho$ by $\frac{\mathbb I}{n} + \Lambda \cdot \vec r$ and expanding $Tr(\rho^k)$.

If we are interested in identifying vectors $\vec r$ that correspond to pure state density matrices, we only need to check the three conditions, $Tr(\rho)=1,Tr(\rho^2)=1$ and $Tr(\rho^3)=1$, since these are satisfied if and only if the density matrix corresponds to a pure state as proven in \cite{appleby2007symmetric}. Finally, we write the explicit functions that determine pure states on the Bloch sphere as 

\begin{equation}
\begin{array}{ll}
     Tr(\rho^2) &= Tr\big(\frac{\mathbb I}{n^2} + \frac{2}{n}\Lambda \cdot \vec r +(\Lambda \cdot \vec r)^2\big)\\ 
     &= \frac{1}{n}+2\vec r \cdot \vec r \mbox{ .}
\end{array}
\end{equation}

\noindent
Similarly, for the cube case, we write the following,

\begin{equation}
\begin{array}{ll}
     Tr(\rho^3) &= Tr\big(\frac{\mathbb I}{n^3} + \frac{3}{n^2}\Lambda \cdot \vec r +\frac{3}{n}(\Lambda \cdot \vec r)^2+(\Lambda \cdot \vec r)^3\big)\\ 
     &= \frac{1}{n^2}+\frac{6}{n^2}\vec r \cdot \vec r + \sum_{i,j,k=1}^{n} Tr(\Lambda_i\Lambda_j\Lambda_k) r_ir_jr_k \mbox{ .}
\end{array}
\end{equation}

The triple products $Tr(\Lambda_i\Lambda_j\Lambda_k)$, known as the structure factor of the Gell-Mann matrices, will be denoted by the tensor $d_{ijk}$. After further simplifications, we the two characteristic equations $Tr(\rho^2)=1$ and $Tr(\rho^3)=1$ as

\begin{equation}\label{trace square function}
     |\vec r| = \sqrt{\frac{n-1}{2n}} 
\end{equation}

\noindent
and

\begin{equation}\label{trace functions for purity}
     \sum_{i,j,k=1}^{n^2-1} d_{ijk} r_ir_jr_k = \frac{(n-1)(n-2)}{n^2} \mbox{ }
\end{equation}

\noindent
, respectively.

\subsection{SIC-POVMs on the Bloch Sphere}

Consider the set of $n^2$ normalized vectors $\{|\psi_k\rangle\}$ in $n$-dimensional Hilbert space, where $|\langle\psi_j|\psi_k\rangle|^2=\frac{1}{n+1}$. The corresponding SIC-POVM represented by the vector set is $\{\Pi_k\} =\{\frac{1}{n} |\psi_k\rangle\langle\psi_k|\}$. On the Bloch sphere, these vectors are mapped to $\{\vec v_k\}$ where each of the vectors has a magnitude of $\sqrt{\frac{n-1}{2n}}$, such that 

\begin{equation}\label{SIC on Bloch sphere}
    |\psi_k\rangle\langle\psi_k| = \frac{\mathbb I}{n} +\vec \Lambda \cdot \vec v_k\mbox{ .}
\end{equation}

\noindent
The vectors $\{\vec v_k\}$, like the SIC-POVMs are equiangular. This is derived using \eref{SIC on Bloch sphere} and expanding the expression $Tr(|\psi_j\rangle\langle\psi_j||\psi_k\rangle\langle\psi_k|)=\frac{1}{n+1}$ as  

\begin{equation}
\begin{array}{ll}
Tr(|\psi_j\rangle\langle\psi_j|&|\psi_k\rangle\langle\psi_k|) = Tr(\bigg(\frac{\mathbb I}{n} + \vec v_j \cdot \vec \Lambda\bigg)\bigg(\frac{\mathbb I}{n} + \vec v_k \cdot \vec \Lambda\bigg))\\
&=\frac{1}{n^2}Tr(\mathbb I) + \frac{(\vec v_j+\vec v_k)}{n}.Tr(\vec \Lambda)+ Tr((\vec v_j \cdot \vec \Lambda) (\vec v_k \cdot \vec \Lambda))\\
&= \frac{1}{n}+2\vec v_j \cdot \vec v_k = \frac{1}{n+1} \mbox{ .}
\end{array}
\end{equation}

\noindent
Finally, $\vec v_j \cdot \vec v_k =-\frac{1}{2n (n+1)}$ and the angle between any two vectors is $\cos(\theta_{jk}) = -\frac{1}{n^2-1}$. This shows that the vectors $\{\vec v_k\}$ are the vertices of a regular $(n^2-1)$-simplex. 

In dimension $2$, where all points on the surface of the Bloch sphere correspond to pure state density matrices, a SIC-POVM can be formed by any regular $3$-simplex. In higher dimensions, one can not arbitrarily orient the simplex and form a SIC-POVM. This is because the surface of the Bloch sphere contains points that do not correspond to any physical density matrix and the pure states are the global maxima of the polynomial function shown in \eref{trace functions for purity} on the surface of the Bloch sphere.

By defining the map $f: \mathbb R^{n^2-1} \mapsto \mathbb R$ given by
\begin{equation}\label{trace cube}
f(\vec r) = \sum_{ijk} d_{ijk}r_ir_jr_k \mbox{ ,}
\end{equation}
\noindent
we restate the construction of a SIC-POVM in $n$-dimensional Hilbert space as a problem of orienting a $(n^2-1)$ dimensional regular simplex such that all the vectors $\{\vec v_k\}$ are mapped to the same value of $f(\vec v_k) = \frac{(n-1)(n-2)}{n^2}$. Defining the SIC-POVM in this manner allows us to see the connection between the SIC-POVMs and Knaster's conjecture, which is discussed in the next section.

\section{Proof of the Knaster's Conjecture for $n$ Vertices of an $n$-Simplex}\label{proof of conjecture}

We start by constructing a regular simplex in the $n^2-1$ dimensional real space. This can be done using the Gram-Schmidt process. Start with the one dimensional vectors $V_1=\{(1),(-1)\}$, which is the equiangular vector in the $1$-dimensional vector space. Then, by applying the algorithm $V_n = \{\{\frac{\sqrt{n^2-1}}{n}V_{n-1},\frac{1}{n}\},\{0,...,0,-1\}\}$, we construct $n+1$ equiangular unit vectors. For example, in the $2$-dimensional vectors space, the set of equiangular vectors is $V_2=\{(\frac{\sqrt{3}}{2},\frac{1}{2}),(-\frac{\sqrt{3}}{2},\frac{1}{2}),(0,-1)\}$ and in the $3$-dimensional vector space, the corresponding set is $V_3=\{(\frac{\sqrt{8}}{3}V_2,\frac{1}{3}),(0,0,-1)\}$. By plugging the vectors of $V_2$, we generate the four vectors in $V_3$, which gives the vectors 

\begin{equation}\label{3 simplex}
     \{(\frac{\sqrt{6}}{3},\frac{\sqrt{2}}{3},\frac{1}{3}), (-\frac{\sqrt{6}}{3},\frac{\sqrt{2}}{3},\frac{1}{3}),
    (0,-\frac{\sqrt{8}}{3},\frac{1}{3}),(0,0,-1)\} \mbox{ .}
\end{equation}

The vertices of a regular simplex and a set orthonormal basis vectors have geometric similarities which simplify the proof of Knaster's conjecture. In both cases, orthogonal transformations can be constructed such that any number of chosen vectors in the set are fixed by the transformations and the rest of the vectors form a connected subspace as shown in lemma \ref{fixing linealy independent vectors}.  

\begin{lemma}\label{fixing linealy independent vectors}
Consider a set of vectors $P = \{\vec{p}_1,...,\vec{p}_n\}$ on $S^{n-1}$ where $\vec{p}_i.\vec{p}_j = -\frac{1}{n}$, then there exists a sets of matrices $\{U_1,...,U_{n-1}\}$ where $U_i \in SO(n)$ such that
\begin{equation}
\begin{array}{ll}
U_1\vec{p}_1 &= \vec{p}_1 \mbox{ ,}\\
U_2\vec{p}_2 &= \vec{p}_2 , U_2\vec{p}_1 = \vec{p}_1  \mbox{ ,}\\
U_3\vec{p}_3 &= \vec{p}_3 , U_3\vec{p}_2 = \vec{p}_2 , U_3\vec{p}_1 = \vec{p}_1 \mbox{ ,}\\
\vdots\\
U_{n-1}\vec{p}_{n-1} &= \vec{p}_{n-1} , ... , U_{n-1}\vec{p}_2 = \vec{p}_2  , U_{n-1}\vec{p}_1 = \vec{p}_1 \mbox{ .} 
\end{array}
\end{equation}

\end{lemma}

\begin{proof}
    
Let's take the mutually orthogonal vectors $\{\vec{v}_1,...,\vec{v}_n\}$ constructed from the elements of $P$ using the Gram-Schmidt algorithm. Then the elements of $P$ take the following form,

\begin{equation}
\begin{array}{ll}
\vec{v}_1 = \vec{p}_1 \mbox{ ,}\\
\vec{v}_2 +\frac{(\vec{v}_1.\vec{p}_2)}{|\vec{v}_1|^2}\vec{v}_1 = \vec{p}_2 \mbox{ ,}\\
\vec{v}_3 +\frac{(\vec{v}_2.\vec{p}_3)}{|\vec{v}_2|^2}\vec{v}_2 + \frac{(\vec{v}_1.\vec{p}_3)}{|\vec{v}_1|^2}\vec{v}_1 = \vec{p}_3 \mbox{ ,}\\
\vdots\\
\vec{v}_n +\frac{(\vec{v}_{n-1}.\vec{p}_n)}{|\vec{v}_{n-1}|^2}\vec{v}_{n-1} + ... + \frac{(\vec{v}_2.\vec{p}_n)}{|\vec{v}_2|^2}\vec{v}_2 + \frac{(\vec{v}_1.\vec{p}_n)}{|\vec{v}_1|^2}\vec{v}_1 = \vec{p}_n \mbox{ .}
\end{array}
\end{equation}

\noindent
By using vector set $\{\hat{v}_1,...,\hat{v}_n\}$ as our basis, we construct the set of matrices $U_i$ as 
\begin{equation}\label{selective unitary}
U_i=\{\left[
   \begin{array}{cc}
   I & 0  \\
   0 & u_k   \\
   \end{array} \right]
   : \forall u_i \in SO(n-i)\} \mbox{ .}
\end{equation}

\noindent
where $I$ is a $k$-dimensional identity matrix. The action of the matrices $U_k$ on the equiangular vectors $P$ fixes the first $k$ vectors. In addition, matrices of the form $U_k$ exist which map each of the remaining vectors to each other. This shows that the orbits of these vectors are equivalent.

\end{proof}

Given the equiangular vectors, $V_{n^2-1}$, and the corresponding orthogonal matrices $U_k$ which fixed the first $k$ vectors, the orbit of the remaining $(n-k)$ vectors is defined as the collection of points having a fixed inner product with all the $k$ vectors. The subspaces can be defined as the intersection of the subspaces with a fixed inner product to each vector of index $1$ up to $k$. Before defining these subspaces, let us define a representation for a surface. A surface defined by a function $g(\vec x)-g_0=0$ is represented as $\mathcal S = \{\vec x \in \mathbb R^N | g(\vec x) - g_0=0\}$. we will represent such surface in a slightly shorter form $\{g(\vec x)-g_0=0\}_N$.

\begin{define}\label{series of j}
Given the set of equiangular vectors $\{\vec{p}_1,\vec{p}_2,...,\vec{p}_n\}$ on $S^{n-1}$ where $\vec{p}_i.\vec{p}_j = -\frac{1}{n}$, define the subspace $J_i = S^{n-1}\bigcap^{i-1}_{j=1}{\{\vec{x}.\vec{p}_j+\frac{1}{n}=0\}}_n$ and $J_1 = S^{n-1}$.
\end{define}

The subspaces $J_i$ correspond to the orbits formed by the matrices $U_{i-1}$. This can be understood from the fact that the inner product of any vector in $J_i$ and the vectors $\{\vec p_1,\cdots,\vec p_{i-1}\}$ is  $\frac{1}{n}$. A key point to note from the definition is that $J_1 \supset J_2 \supset ...\supset  J_n $, which will be important in the proof of theorem \ref{kakutani theorem}. Next, we show that each of the surfaces $J_i$ is isomorphic to the sphere $S^{n-i}$ using proof by induction. 

\begin{lemma}\label{J isomorphism}
$J_i \cong S^{n-i}:1 \leq i \leq n$

\begin{proof}

For $i=1$, by definition $J_1=S^{n-1}$.

\noindent
For $i=2$, the set is given by,

\begin{equation}\label{J2 surface}
J_2 = S^{n-1} \cap \{\vec{x}.\vec{p}_1+\frac{1}{n} = 0\}_n \mbox{ .}
\end{equation}

\noindent
Let $\{\hat{e}_1,\dots,\hat{e}_{n-1}\}$ be a set of orthogonal vectors such that $\hat{e}_i.\vec{p}_1 = 0$. Then, we can write the vector $\vec{x}$ as,

\begin{equation}
\vec{x} = \sum_{i=1}^{n-1}{\alpha_i\hat{e}_i}-\frac{\vec{p}_1}{n} , \alpha_i \in \mathbb R \mbox{ .}
\end{equation}

\noindent
This form of $\vec{x}$ satisfies the equation $\vec{x}.\vec{p}_1+\frac{1}{n} = 0$. The vector $\vec x$ is also an element of $S^{n-1}$, meaning $|\vec x|^2 = 1$. Using the normalization conditions of $\vec x$ and $\vec p_1$ in 

\begin{equation}
|\vec{x}|^2 = \sum_{i=1}^{n-1}{\alpha_i^2}+\frac{|\vec{p}_1|^2}{n^2} \mbox{ ,}
\end{equation}

\noindent
we calculate the sum of $\alpha_i^2$ as 

\begin{equation}\label{s(n-2)}
\sum_{i=1}^{n-1}{\alpha_i^2} = \big(\frac{n^2-1}{n^2}\big) \mbox{ .}
\end{equation}

\noindent
Clearly, \eref{s(n-2)} defines a sphere $S^{n-2}$, and we can form a homeomorphism between the subspace $J_2$ and $S^{n-2}$ defined as $\phi_2 : J_2 \mapsto S^{n-2}, \phi_2(\vec x) = \vec{x}+\frac{\vec{p}_1}{n}$ which proves $J_2 \cong S^{n-2}$.

For $i=3$, we start by writing the subspace $J_3$ and an intersection of $J_2$ and $\{\vec{x}.\vec{p}_2 + \frac{1}{n} = 0\}$, which becomes

\begin{equation}
J_3 = J_2 \cap \{\vec{x}.\vec{p}_2 + \frac{1}{n} = 0\} \mbox{ .}
\end{equation}

By applying the map $\phi_2$ on the subspace $J_3$, we generate the subspace $\phi_2(J_3)$ which is homeomorphic to the original space $J_3$. 
The resulting subspace of $S^{n-1}$ is $\phi_2(J_3) = \phi_2(J_2 \cap \{\vec{x}.\vec{p}_2 + \frac{1}{n} = 0\}) = \phi_2(J_2) \cap \phi_2(\{\vec{x}.\vec{p}_2 + \frac{1}{n} = 0\})$. The subspace $\phi_2(J_2) $ is a $n-1$ dimensional sphere of radius $\sqrt{\frac{n^2-1}{n^2}}$. Next, we need to write $J_3$ in the same form as $J_2$, specifically, as an intersection of a sphere and a surface defined by some continuous equation. This is done by applying the map $\phi_2$ on the vectors $\vec x$ and $\vec p_2$. Let's define $\vec{x}' = \vec{x} + \frac{\vec{p}_1}{n}$ and $\vec{p'}_2 = \vec{p}_2 + \frac{\vec{p}_1}{n}$. By replacing the vectors in the surface equation $\vec{x}.\vec{p}_2 + \frac{1}{n} = 0$ by the vectors $\vec x'$ and $\vec{p'}_2$, we form the subspace $\phi_2(J_3)$ as

\begin{equation}\label{transformed J3}
\begin{array}{ll}
\phi_2(J_3) &= \phi_2(J_2) \cap \big\{(\vec{x'} - \frac{\vec{p}_1}{n})\cdot(\vec{p'}_2 - \frac{\vec p_1}{n}) + \frac{1}{n} = 0\big\}\\
&= \phi_2(J_2) \cap \big\{(\vec{x'}\cdot\vec{p'}_2 - \frac{\vec{p}_1}{n}\cdot\vec{p'}_2 -\vec{x'}\cdot\frac{\vec{p}_1}{n}+\frac{\vec{p}_1}{n}\cdot\frac{\vec{p}_1}{n}) + \frac{1}{n} = 0\big\} \mbox{ .}
\end{array}
\end{equation}

\noindent
Since $\vec{p}_1.\vec{p'}_2 = 0 $ and for $\vec x \in J_2$ we have $\vec{p}_1.\vec{x'} = 0$, \eref{transformed J3} simplifies to 

\begin{equation}
\begin{array}{ll}
\phi_2(J_3) &= \phi_2(J_2) \cap \big\{\vec{x'}.\vec{p'}_2 +\frac{1}{n^2} + \frac{1}{n} = 0\big\}\\
&= \phi_2(J_2) \cap \big\{\vec{x'}.\vec{p'}_2 +\big(\frac{n+1}{n^2}\big) = 0\big\}\\
&= \phi_2(J_2) \cap \big\{\vec{x'}.\vec{p'}_2 +\big(\frac{n^2-1}{n^2}\big)\frac{1}{n-1} = 0\big\} \mbox{ .}
\end{array}
\end{equation}

\noindent
The magnitude of the  vectors $\vec x'$ and $\vec{p'}_2$ is $\sqrt{\frac{n^2-1}{n^2}}$ as is the radius of the sphere $\phi_2(J_2)$. Thus, we define a normalization map $r_2: R^{n-1} \mapsto R^{n-1}, r_2(\vec{x'}) = \sqrt{\frac{n^2}{n^2-1}}\vec{x'} $ and apply it on $\phi_2(J_3)$. The map $r_2(\phi_2(J_2))$ simply gives $S^{n-2}$ and for the equation of the surface, we replace $\vec x'$ and $\vec{p'}_2$ by $r(\vec{x'})\sqrt{\frac{n^2-1}{n^2}}$ and $r(\vec{p'}_2)\sqrt{\frac{n^2-1}{n^2}}$, respectively. The resulting subspace is, 

\begin{equation}\label{reduced j1}
(r_2 \circ \phi_2)(J_3) = S^{n-2} \cap \big\{r(\vec{x'})\cdot r(\vec{p'}_2) +\frac{1}{n-1} = 0\big\} \mbox{ .}
\end{equation}

Since $r_2$ is a homeomorphism, the map $r_2 \circ \phi_2$ is also a homeomorphism. Note that $(r_2 \circ \phi_2)(J_3)$ has the same form as $J_2$ with one less dimension. Therefore, we can go through the same steps as $J_2$ to form a homeomorphism $\phi_3: (r_2 \circ \phi_2)(J_3) \mapsto S^{n-3}, \phi_3(x) = \vec{x} + \frac{\vec{p'}_2}{n-1} $, thereby showing $J_3 \cong S^{n-3}$. Finally, we repeat the process for each surface $J_i$ to prove that $J_i \cong S^{n-i}$ for all $1 \leq i \leq n$.

\end{proof}

\end{lemma}

Finally, we are ready to show that, given any continuous function on $S^n$, there exists a regular simplex for which at least $n$ of its vertices are mapped to a unique value by the function. 

\begin{theorem}\label{kakutani theorem}
Let $f$ be a continuous function defined on the unit sphere $S^{n-1}$ such that $f : S^{n-1} \mapsto \mathbb R$. There exists a set of points $\{p_1,p_2,...,p_n\}$ on $S^{n-1}$ such that,

\begin{equation}
\vec{p}_i.\vec{p}_j = -\frac{1}{n}  
\end{equation}
and
\begin{equation}
f(p_1) = f(p_2) = ... = f(p_n) = f_0 \mbox{ , } f_0 \in \mathbb R \mbox{ ,}
\end{equation}

\noindent
where $\vec{p}_i$ represents the vector from the origin to the point $p_i$.

\end{theorem}

\begin{proof}
We start by noting that, since $f$ is a continuous function and $S^{n-1}$ is a compact space, the value of $f(p)$, $\forall p \in S^{n-1}$ is continuous and bounded, i.e. $a \leq f(p) \leq b$. Let's take two points $p_a$ and $p_b$ where $p_a,p_b \in S^{n-1}$, such that $f(p_a)=a$ and $f(p_b)=b$ and define a set of matrices $\mu(\theta) = \{\mu(\theta) \in SO(n), \mu(1)\vec{p_a}=\vec{p_b}, \theta \in [0,1] \}$. We then construct a set of equiangular vectors $\Psi = \{p_1,p_2,...,p_n\}$ on $S^{n-1}$ by using the method shown in beginning of this section, such that $\vec{p}_i.\vec{p}_j = -\frac{1}{n}$ and $p_1 = p_a$. When we apply $\mu(\theta)$ on $\Psi$ we get the set $\Psi(\theta) = \{p_1(\theta),p_2(\theta),...,p_n(\theta)\}$ and since $\mu(\theta) \in SO(n)$ the angle between the vectors $\vec{p}_i(\theta)$ are invariants.

First, let's assume that the theorem holds for any continuous function defined on $S^{n-2}$. Then we can construct a set of $n-1$ points $\{r_1,...,r_{n-1}\}$ where $\vec{x}_i.\vec{x}_j = -\frac{1}{n-1}$, such that $\forall k, f(\vec r_k)=f_0$. For a given $p_1(\theta)$, we construct a subspace $J_2$ defined in def \ref{series of j}. As shown in lemma \ref{J isomorphism}, $J_2 \cong S^{n-2}$. By using the map $\phi_2$ introduced in the proof of lemma \ref{J isomorphism}, we can map the vectors $p_k,2\leq k \leq n$, to unit length vectors in $(n-1)$-dimensional sphere, where

\begin{equation}
    \phi_2(\vec p_j)\cdot\phi_2(\vec p_k) = \bigg(-\frac{1}{n-1}\bigg)|\phi_2(\vec p_j)||\phi_2(\vec p_k)| \mbox{ .}
\end{equation}

\noindent
The corresponding function on the subspace becomes $f(\vec p_k) \mapsto f\circ \phi^{-1}(\vec{p'}_k)$, where $\vec{p'}_k = \phi_2(\vec p_k)$. Based on our assumption, there exists an orthogonal matrix $u_1\in SO(n-2)$ such that $f\circ\phi^{-1}(u_1\vec{p'}_k) = f_0$ is a constant for all $k\neq 1$. The orthogonal matrix $u_1$ can be used to form an orthogonal matrix $U_1$ in $SU(n-1)$ as shown in \eref{selective unitary}. In other words, by taking $\mu(\theta) = U_1$, we construct a set of vectors $\Psi(\theta) = \{p_1(\theta),p_2(\theta),...,p_n(\theta)\}$, where $f(p_2(\theta)) = f(p_3(\theta)) = ... = f(p_n(\theta)) = f_\theta$. Since $f(p_1(\theta))$ goes from $a$ to $b$ continuously as $\theta$ goes from 0 to 1 and $\forall \theta$ $a \leq f_\theta \leq b$, there must exist a point $\theta$ for which $f(p_1(\theta)) = f(p_2(\theta)) =...= f(p_n(\theta)) = f_\theta$. We can then conclude that, if the theorem holds for an arbitrary continuous function on $S^{n-2}$, it must also be true for an arbitrary function on $S^{n-1}$.

Finally, we need to show the theorem holds for $S^1$ to complete the proof. For $S^1$, we define the set of points $\Psi(\theta) = \{p_1(\theta),p_2(\theta)\}$ where $f(p_1(0)) = a$ and $f(p_1(1)) = b$ and $a \leq f(p_2(\theta)) \leq b$. Since $f$ is a continuous function, there exists a point $\theta$ where $f(p_1(\theta)) = f(p_2(\theta))$. 

This concludes the proof that for any continuous function on $S^{n-1}$, there exists a set of equiangular vectors $\{\vec{p}_1,\vec{p}_2,...,\vec{p}_n\}$ on $S^{n-1}$ where $\vec{p}_i.\vec{p}_j=-\frac{1}{n}$ and $f(\vec p_1)=f(\vec p_2)=...=f(\vec p_n)$.
\end{proof}

Extending the theorem to all of the $n+1$ vertices of the simplex requires further analysis of the specific function in question. One simple example is the function $f(\vec r) = 3x^2y-y^3$ on $S^1$. Since the function has $D_6$ symmetry group, we can generate a $2$-simplex where all $3$ vertices map to the same value of $f(\vec r)$. Another general class of functions are ones where $f_0$ takes discrete values. Such functions necessarily form a subspace $J_n$ where all the vectors map to $f_0$. As a result, $n+1$ of the vertices will map to the same value $f_0$.

\subsection{Generalized SIC-POVMs}

As shown in \eref{trace cube}, a vector $\vec r$ on the surface of the Bloch sphere corresponds to a pure state density matrix if and only if $\sum_{ijk} d_{ijk}r_ir_jr_k = \frac{(n-1)(n-2)}{n^2}$. Let the function $f(\vec r)$ be the polynomial function $ \sum_{ijk} d_{ijk}r_ir_jr_k$. Since the function is continuous, theorem \ref{kakutani theorem} shows that in the Bloch sphere of the $n$-dimensional Hilbert space, a regular $(n^2-1)$-simplex exists such that all but one of its vertices satisfy $f(\vec p_k) = f_0$ for some constant $f_0 \in [-\frac{(n-1)(n-2)}{n^2},\frac{(n-1)(n-2)}{n^2}]$. A SIC-POVM on the other hand corresponds to a set of $n^2$ vectors $\vec r_k$ on the Bloch sphere, where $f(\vec r_k) = \frac{(n-1)(n-2)}{n^2}$ for all $n^2$ vectors. Consequently, the SIC existence problem can be equivalently stated as the following two questions. First, can the value $f_0$ take the maximum value of $\frac{(n-1)(n-2)}{n^2}$? Second, is the last vertex of the simplex, which can be written in its matrix form as

\begin{equation}\label{last vertex of the simplex}
    \Pi_{n^2} = n\mathbb I - \sum_k^{n^2-1} \Pi_k \mbox{, where } \Pi_k = \frac{1}{n}\mathbb I + \vec \Lambda \cdot \vec p_k \mbox{ ,}
\end{equation}
\noindent
a pure state?

Regarding the first question, the function $f(\vec r)$ was chosen to be the simplest among all the functions derived from $Tr(\rho^k)$. In principle, the question can be expressed for any of the $\{Tr(\rho^k), k>2\}$ functions, which has the maximum value if $\rho$ is a pure state density matrix. 

In general, theorem \ref{kakutani theorem} shows that some orientation of the simplex exists where all the vertices correspond to density matrices with the same $Tr(\rho^k)$ value. In dimensions greater than 3, the existence of SIC-POVMs requires that the theorem \ref{kakutani theorem} be valid for continuous values of $f_0$. This result is particularly interesting in the $3$-dimensional Hilbert space, where $Tr(\rho^2)$ and $Tr(\rho^3)$ uniquely identify a density matrix up to a unitary transformation. If the value of $f_0$ is a maximum value of the trace functions, then there exist continuous orientations of the simplex for continuous values of $f_0$, as shown in theorem \ref{continuity proof}. In dimension $3$, the continuous value of $f_0$ means that a set of $8$ density matrices, which are equivalent up to unitary transformations can be constructed. We confirmed the results numerically and found that, in dimensions 3 and 4, simplexes exist such that all the $n^2$ vertices are mapped to some values $f_0$ on the functions $Tr(\rho^3)$ and $Tr(\rho^4)$.

\begin{theorem}\label{continuity proof}
    Let $V$ be a set of $(n^2-1)$ vectors $\{\vec p_k\}$ on the Bloch sphere of an $n$-dimensional Hilbert space where $n\geq3$, such that $\vec p_j \cdot \vec p_k = -\frac{1}{n^2-1}|\vec p_j||\vec p_k|$ and $\forall k, f(\vec p_k) = f_0$, for the continuous function $f$ given in \eref{trace cube}. Then, a SIC-POVM exists in $n$-dimensional Hilbert space only if the set of vectors $V$ exists for a continuous value of $f_0$. 
\end{theorem}

\begin{proof}
We start by assuming that theorem \ref{kakutani theorem} holds for a discrete value of $f_0 = \frac{(n-1)(n-2)}{n^2}$, Which is the necessary condition for purity of a density matrix. Let the set of $n^2$ vectors $\{\vec p_k\}$ form the SIC-POVM. Consider the subspace $J_{n^2-2}$ generated by the vectors $\{\vec p_k\}$, which as shown in \ref{J isomorphism} is isomorphic to $S^1$. Based on our assumption, $f(\vec x) = f_0$ for all vector $\vec x \in J_{n^2-2}$. The subspace $J_{n^2-2}$ contains the last $3$ vectors of the set. Let the three density matrices be $\{\rho_a,\rho_b,\rho_c\}$, which correspond to the vectors $\{\vec p_{n^2-2},\vec p_{n^2-1},\vec p_{n^2}\}$ on the Bloch sphere, respectively. 

Next, we will construct the circle $J_{n^2-2}$ by using the vectors. For computational simplicity, we will use the Gell-Mann matrices as the basis vectors. This will allow us to write the vectors as matrices, where the dot product of two vectors is replaced with a trace.

The subspace $J_{n^2-2}$ can be constructed by using the three vectors as follows. Let the three vectors $\vec r_1,\vec r_2,\vec r_3$ represent the vectors $\{\vec p_{n^2-2},\vec p_{n^2-1},\vec p_{n^2}\}$. Define the midpoint vector of the three vectors as $r_0=\frac{1}{3}(\vec r_1+\vec r_2+\vec r_3)$ and the vectors connecting $\vec r_0$ to $\vec r_k$ as $\vec v_k$, i.e., $\vec v_k = \vec r_k -\vec r_0$. The subspace $J_{n^2-1}$ can then be written as some linear combination of the basis vectors $\vec e_1$ and $\vec e_2$ given by the expressions 

\begin{equation}
   \begin{array}{ll}
    \vec e_1 & = \vec v_1  \mbox{ ,}\\
    \vec e_2 & = \vec v_2 -\frac{\vec v_2.\vec v_1}{|\vec v_1|^2}\vec v_1 \mbox{ .}
   \end{array}
\end{equation}

After normalizing $\vec e_2$ to have the same magnitude as $\vec e_1$, we can define the circle as $\cos (\theta)\vec e_1+\sin (\theta) \vec e_2$. First, we express the vectors $\vec v_k$ using the matrix representation by using the Gell-Mann matrices as $\vec r_k \rightarrow \vec r. \vec \Lambda =  \rho_k - \frac{1}{n}\mathbb I$ as 

\begin{equation}
\begin{array}{ll}
    \vec v_1 & \rightarrow \frac{2}{3} \rho_a - \frac{1}{3} (\rho_b+\rho_c) \mbox{ ,}\\
    \vec v_2 & \rightarrow \frac{2}{3} \rho_b - \frac{1}{3} (\rho_a+\rho_c) \mbox{ ,}\\
    \vec v_3 & \rightarrow \frac{2}{3} \rho_c - \frac{1}{3} (\rho_a+\rho_b) \mbox{ .}\\
\end{array}
\end{equation}

\noindent
We form the vectors $\vec e_1$ and $\vec e_2$ similarly, where $\vec v_2.\vec v_1 = \frac{1}{2}Tr\big((\frac{2}{3} \rho_1 - \frac{1}{3} (\rho_2+\rho_3))(\frac{2}{3} \rho_2 - \frac{1}{3} (\rho_1+\rho_3))\big)$ and  $\vec v_1.\vec v_1 = \frac{1}{2}Tr\big((\frac{2}{3} \rho_1 - \frac{1}{3} (\rho_2+\rho_3))^2\big)$, which becomes 

\begin{equation}
    \begin{array}{ll}
    \vec e_1 & \rightarrow \frac{2}{3} \rho_a - \frac{1}{3} (\rho_b+\rho_c) \mbox{ ,}\\
    \vec e_2 & \rightarrow \frac{1}{\sqrt{3}}(\rho_b-\rho_c) \mbox{ .}
    \end{array}
\end{equation}

\noindent
Then, we write the circle $J_{n^2-2}$ as $\cos (\theta) \vec e_1 + \sin (\theta) \vec e_2 + \vec r_0$, which has the following form in the matrix representation,

\begin{equation}\label{J subspace matrix form}
\begin{array}{ll}
    \Omega = (\frac{2}{3}\cos (\theta) + \frac{1}{3}) \rho_a + (-\frac{1}{3}\cos (\theta) +\frac{1}{\sqrt{3}}\sin (\theta) + \frac{1}{3})\rho_b \\ + (-\frac{1}{3}\cos (\theta) -\frac{1}{\sqrt{3}}\sin (\theta) + \frac{1}{3})\rho_c - \frac{1}{n}\mathbb I \mbox{ .}
\end{array}
\end{equation}

\noindent
If $f(\vec x) = \frac{(n-1)(n-2)}{n^2}$ for all vectors in $J_{n^2-2}$, then the matrix shown in \eref{J subspace matrix form} must satisfy $Tr\big((\frac{1}{n} \mathbb I + \Omega)^3\big) = 1$ for all $\theta$. We expand the trace of $(\frac{1}{n} \mathbb I + \Omega)^3$ as

\begin{equation}\label{continuity condition on J}
\begin{array}{ll}
     Tr\big((\frac{1}{n} \mathbb I + \Omega)^3\big) = \frac{1}{9 (n+1)^{3/2}} &\bigg( -4 \alpha +2 \cos (3 \theta ) \left(2 \alpha +\sqrt{n+1} n-2 \sqrt{n+1}\right) \\
     & +7 \sqrt{n+1} n+13 \sqrt{n+1}\bigg) \mbox{ ,}
\end{array}
\end{equation}

\noindent
where $\alpha$ is the cosine of the phase of the triple product $Tr(\rho_a\rho_b\rho_c)$. From \eref{continuity condition on J}, the coefficient of $\cos (3\theta)$ vanishes for $\alpha = \frac{1}{2}(\sqrt{n+1} n-2 \sqrt{n+1})$. Since $-1 \leq \alpha \leq 1$, $n$ can only be $3$, where $\alpha = -1$. Therefore, for dimensions greater than $3$, $f(\vec x)$ takes a continuous value for $\vec x \in J_{n^2-2}$  and as a result, theorem \ref{kakutani theorem} holds for a continuous value of $f_0$.

In dimension $3$, we can use the identity presented in corollary 2 of the article \cite{appleby2011lie},

\begin{equation}\label{triple product identity}
    \sum_{rst} Tr(\Pi_r\Pi_s\Pi_t) = n^4 \mbox{ ,}
\end{equation}

\noindent
to show that all the phases of the triple product can not be $\pi$. We can then choose three vectors of the SIC-POVMs that satisfy $Tr(\rho_a \rho_b \rho_c) \neq \pi$ to show that $f_0$ takes continuous values, thereby concluding the proof.
\end{proof}

Theorem \ref{continuity proof} can be extended to all $Tr(\rho^k)$ functions for $k>3$. To show this, note that pure states can be identified by any one of the functions for powers $k\geq3$, as shown in section \ref{density matrix on bloch}. As a result, if a subspace $J_{n^2-2}$ containing only pure states doesn't exist, then the functions $Tr(\rho^k)$ cannot be $1$ for any subspace $J_{n^2-1}$ as well. Consequently, continuous generalized SIC-POVMs must exist for any function $Tr(\rho^k)$.

In both dimensions, we searched for generalized SIC-POVMs, such that all of the elements have the same $Tr(\rho^3)$ or $Tr(\rho^4)$. In dimension 3, we generated $10^4$ solutions with precision of $O(10^{-18})$ by starting from a known SIC-POVM and minimizing the function $\sum_k^{n^2}(f(\vec p_k)-f_0)^2$ for arbitrary values $f_0 \in [-\frac{2}{9},\frac{2}{9}]$, where $f(\vec r)$ is the polynomial 

\begin{equation}\label{trace cube in 3d}
\begin{array}{ll}
    f(\vec r)&=2 \sqrt{3} r_1^2 r_8+6 r_1 r_2 r_3+6 r_1 r_5 r_6+3 r_2^2 r_7-\sqrt{3} r_2^2 r_8 \\ &-6 r_2 r_4 r_6-3 r_3^2 r_7-\sqrt{3} r_3^2 r_8+6 r_3 r_4 r_5
    +2 \sqrt{3} r_4^2 r_8\\&+3 r_5^2 r_7-\sqrt{3} r_5^2 r_8-3 r_6^2 r_7-\sqrt{3} r_6^2 r_8+2 \sqrt{3} r_7^2 r_8-\frac{2 r_8^3}{\sqrt{3}} \mbox{ .}
\end{array}
\end{equation}
Similarly, we generated general SIC-POVMs $\{\Pi_k\}$ such that $Tr(\Pi_k^3)=f_0$ and $Tr(\Pi_k^4)=f_0$ separately. However, unlike in dimension 3,  we weren't able to construct general SIC-POVMs where all elements have the same set of eigenvalues.

\section{Conclusion}

A geometric approach to SIC-POVMs is an ideal tool for the existence problem. By mapping the SIC-POVMs onto a real space, we showed that the problem of the existence of SIC-POVMs is related to Knaster's conjecture, which is a much broader geometric problem. The symmetry of the SIC-POVMs on the Bloch sphere, i.e., the simplex, is used to prove the conjecture for $N$ vectors of a regular $N$-simplex. As a result, we concluded that, In $n$-dimensional Hilbert space, a generalized SIC-POVM can be constructed such that $n^2-1$ of its elements have the $Tr(\rho^m)$ value for $m\geq3$. In dimension $3$, where the only relevant continuous function is formed from $Tr(\rho^3)$, generalized SIC-POVMs can be constructed such that 8 of the elements are equivalent up to some unitary operations. From the existence of the SIC-POVMs point of view, the result motivates the following questions. The first one is, Can we prove that in $n$-dimensional Hilbert space, a general SIC-POVMs with $n^2-1$ rank-1 projection operators exist? The second question is, If $n^2-1$ elements of a general SIC-POVM are rank-$1$ projection operators, can we prove that the last element is necessarily a rank-$1$ projection operator as well? We explored the two questions numerically for $3$ and $4$ dimensional Hilbert spaces. In both cases, we found that a general SIC-POVM can be constructed where all of its elements have the same value of $Tr(\rho^3)$ and $Tr(\rho^4)$, separately. In dimension 3, the large number of solutions generated suggests that generalized SIC-POVMs can be generated such that its elements have an arbitrary set of eigenvalues. 
 
To summarize, whether theorem \ref{kakutani theorem} can be extended to all $n^2$ vectors in the Bloch sphere of $n$-dimensional Hilbert space remains to be answered. All numerical searches in dimensions $3$ and $4$ have yielded general SIC-POVMs with $n^2$ of the vectors mapping to the same value of $Tr(\rho^3)$. A stronger version of the question is, can we construct $n^2$ equiangular vectors on the Bloch sphere such that all the vectors have the same $Tr(\rho^m)$ value of all $m$? This would generalize the results of dimension 3 to higher dimensions.

\section*{Acknowledgments}

Part of this work was carried out at Bilimler Köyü at Foça.

\section*{Reference}
\bibliography{ref}

\end{document}